\documentclass[pra,twocolumn,floatfix,aps,superscriptaddress,showpacs,amsfonts]{revtex4}
\usepackage{graphicx}
\usepackage{array}
\usepackage{amsthm}
\usepackage{amsmath}
\usepackage{amssymb}
\newtheorem{theorem}{Theorem}

\begin{document}
\title{Measurement-induced nonlocality based on the relative entropy}
\author{Zhengjun Xi}
\affiliation{College of Computer Science, Shaanxi Normal University, Xi'an, 710062, China}
\author{Xiaoguang Wang}
\affiliation{Zhejiang Institute of Modern Physics, Department of Physics, Zhejiang University, Hangzhou 310027, China}
\author{Yongming Li}
\affiliation{College of Computer Science, Shaanxi Normal University, Xi'an, 710062, China}

\begin{abstract}
We quantify the measurement-induced nonlocality (Luo and Fu, Phys.
Rev. Lett. \textbf{106}, 1020401, 2011) from the perspective of the
relative entropy. This quantification leads to an operational
interpretation for the measurement-induced nonlocality, namely, it is
the maximal entropy increase after the locally invariant
measurements. The relative entropy of nonlocality is upper bounded
by the entropy of the measured subsystem. We establish relationship
between the relative entropy of nonlocality and the geometric
nonlocality based on the Hilbert-Schmidt norm, and show that it is
equal to the maximal distillable entanglement. Several trade-off
relations are obtained for tripartite pure states. We also give explicit expressions for
the relative entropy of nonlocality for
Bell-diagonal states.
\end{abstract}
\eid{identifier}
\pacs{03.67.-a, 03.65.Ud, 03.65.Ta}
\maketitle

\section{Introduction}
Nonlocality is a fundamental property of quantum states, and it has
been the subject of intensive studies in the past
decades~\cite{Bergmann,AcinPRA,BrunnerNJP,AugusiakPRL,CavalcantiNC,BancalPRL,BarrettPRL}.
A related topic of interest which has been recently discussed is the measurement-induced
nonlocality~\cite{MethotQIC,FuEPL,GharibianQIC,LuoEPL,LuoPRL},
originating from the superdense coding~\cite{BennettPRL}. One of the
key steps of the superdense coding is to apply a Pauli matrix
operation to half of a Bell state, giving rise to another orthogonal
Bell state, while the reduced density matrices are invariant under
this transformation.
In such a case, the states of both two subsystems are not changed
after the local unitary operation, but the state of the whole system
is changed. This is a nonlocal effect, which was originally
quantified by the Hilbert-Schmidt norm~\cite{FuEPL,GharibianQIC}.
Subsequently, the local unitary invariant was generalized to the
locally invariant measurements, and the measurement-induced
nonlocality was defined via the Hilbert-Schmidt
norm~\cite{LuoEPL,LuoPRL}, which is the so-called geometric
nonlocality.  They derived an analytical formula for any dimensional
pure state and $2\times n$ dimensional mixed states. Recently,
in~\cite{Mirafzali11}, the problem of evaluating
measurement-induced nonlocality for general bipartite mixed states was discussed.

The measurement-induced nonlocality provides a novel classification
scheme for the bipartite states, and may be useful in the
quantitative study of quantum state
steering~\cite{BrunnerNature,WisemanPRL,Schrodinger}.
In this paper, we introduce an entropic measure of nonlocality for
quantum states which is applicable for multipartite systems, and give an physical interpretation. This provides a consistent way
to compare different correlations, such as entanglement, discord,
classical correlations, and quantum
dissonance~\cite{VedralPRAb,Modi2010}.

To quantify the nonlocal effect, we consider the relative entropy between the pre- and post-measurement states,
and introduce the so-called relative entropy of nonlocality.
We then derive some basic properties of the relative entropy of nonlocality.
This quantification leads to an operational interpretation for the nonlocality,
which can be stated by maximal entropy increase in terms of the locally invariant measurement.
 We find that the relative entropy of nonlocality is upper bounded by the von-Neumann entropy of the measured
 subsystem for arbitrary quantum states. As an application, we derive a closed formula for the maximal
 distance achievable for Bell-diagonal states. In such a case, we find that the upper bound of relative entropy of
 nonlocality is saturated for some mixed Bell-diagonal states.

We choose entropic measures for the nonlocality, which is different
from the geometric nonlocality based on the Hilbert-Schmidt
norm~\cite{LuoPRL,LuoEPL}. We can prove that the relative entropy of
nonlocality is always greater than or equal to the squared geometric
nonlocality.
Several trade-off relations can be given by the quantum side
information and the missing information. We also show that the relative entropy of nonlocality
is equal to the maximal distillable entanglement between the
measurement apparatus and the system if the von Neumann measurement
is performed on one part of the system.

This paper is organized as follows. In Sec. \ref{sec:REN}, we define
the relative entropy of nonlocality and present its relevant
properties. We illustrate the relative entropy of nonlocality by the
Bell-diagonal states. In Sec. \ref{sec:REN_HSN}, we discuss the
relationships between the relative entropy of nonlocality and other
measures. We summarize our results in Sec. \ref{sec:conclusion}.

\section{Relative entropy of nonlocality}\label{sec:REN}
Consider a bipartite system with composite Hilbert space $\mathcal{H}=\mathcal{H}^A\otimes\mathcal{H}^B$.
Let $\mathcal{D}(\mathcal{H})$ be the
set of bounded, positive-semidefinite operators with unit trace on $\mathcal{H}$.
Given a quantum state $\rho^{AB}\in \mathcal{D}(\mathcal{H})$ which could be shared between two parties, Alice and Bob, let
$\rho^A$ and $\rho^B$ be the reduced density matrix for each party. As pointed out in~\cite{FuEPL,LuoPRL}, for the bipartite $\rho^{AB}$, one performs local von Neumann measurements which do not disturb the local state $\rho^B=\mathrm{Tr}_A(\rho^{AB})$.
To capture all the nonlocal effects that can be induced by local measurements, Luo defined the measurement-induced nonlocality in terms of the Hilbert-Schmidt norm,
\begin{equation}
\mathcal{N}^{\leftarrow}_{\text{G}}(\rho^{AB})=\max_{\{\Pi^B_k\}}\parallel\rho^{AB}-\tilde{\rho}^{AB}\parallel,\label{def:GNM}
\end{equation}
where the maximum is taken over all the von Neumann measurements $\{\Pi^B_k\}$ which do not disturb $\rho^B$ locally,
that is, $\rho^{B}=\sum_k\Pi^B_k\rho^{B}\Pi^B_k$
and $\tilde{\rho}^{AB}=\sum_k I^A\otimes \Pi^B_k \rho^{AB}I^A\otimes \Pi^B_k$.
The post-measurement state $\tilde{\rho}^{AB}$ can be rewritten as
 \begin{equation}
 \tilde{\rho}^{AB} =\sum_kp_k\rho^A_k\otimes \Pi^B_k,
 \end{equation}
where $\rho_{k}^{A}=\frac{1}{p_k}\mathrm{Tr}_{B}(I^{A}\otimes
\Pi_{k}^{B}\rho^{AB})$ is the post-measurement state of system $A$
that corresponds to the probability $p_{k}=\mathrm{Tr}(I^{A}\otimes
\Pi_{k}^{B}\rho^{AB})$. Here, the Hilbert-Schmidt norm is defined as
$||X||=\sqrt{\mathrm{Tr}(X^\dag X)}$. For the geometric nonlocality~(\ref{def:GNM}), some basic
properties have been listed
in~\cite{LuoPRL,LuoEPL}.

We choose entropic measures for the measurement-induced nonlocality,
and define the relative entropy of nonlocality as
\begin{equation}
\mathcal{N}^{\leftarrow}_{\text{RE}}(\rho^{AB}):=\max_{\{\Pi^B_k\}}S\big(\rho^{AB}||\tilde{\rho}^{AB}\big),\label{eq:REN}
\end{equation}
where the maximum is taken over all the von Neumann measurements $\{\Pi^B_k\}$ which do not disturb $\rho^B$ locally.
Here, $S(X||Y)=\mathrm{Tr}X(\log_2X-\log_2Y)$ is the relative entropy~\cite{Nielsenbook,VedralRMP}.

Note that the relative entropy of nonlocality is different from the relative entropy of quantumness~\cite{Modi2010}
and the one-way quantum deficit~\cite{Oppenheim02,Horodecki062307,ZurekPRA}.
The relationships between the relative entropy of nonlocality and the one-way quantum deficit are
 the same as the relationships between the geometric nonlocality~\cite{LuoPRL} and the geometric measure of quantum discord~\cite{VedralPRL}. Along with this way, the relative entropy of nonlocality is a meaningful and reasonable quantification, and it can be viewed as an indicator of the global effect caused by locally invariant measurements~\cite{LuoPRL,LuoEPL}.

\subsection{Properties for the relative entropy}
We now list some basic properties of the relative entropy of nonlocality $\mathcal{N}^{\leftarrow}_{\mathrm{RE}}$ as follows.

(i) $\mathcal{N}^{\leftarrow}_{\mathrm{RE}}(\rho^{AB})=0$ for any product state $\rho^{AB}=\rho^A\otimes \rho^B$.

(ii) $\mathcal{N}^{\leftarrow}_{\mathrm{RE}}(\rho^{AB})$ is locally
unitary invariant. For any locally unitary operators $U^A\otimes
V^B$ on Hilbert space $\mathcal{H}$, we have
$\mathcal{N}^{\leftarrow}_{\mathrm{RE}}((U^A\otimes
V^B)\rho^{AB}(U^A\otimes
V^B)^\dag)=\mathcal{N}^{\leftarrow}_{\mathrm{RE}}(\rho^{AB})$.

(iii) If $\rho^B$ is non-degenerate, then $\mathcal{N}^{\leftarrow}_{\mathrm{RE}}(\rho^{AB})=S(\rho^{AB}||\tilde{\rho}^{AB})$.

(iv) $\mathcal{N}^{\leftarrow}_{\mathrm{RE}}(\rho^{AB})> 0$ for any entangled state $\rho^{AB}$.

 Following the properties of relative entropy~\cite{VedralRMP}, (i) and (ii) are easy to obtain.
Since the von Neumann measurement $\{\Pi^B_k\}$ does not disturb
$\rho^B$, $\rho^{B}=\sum_k\Pi^B_k\rho^{B}\Pi^B_k$,  after some
manipulation one obtains that each $[\rho^{B},\Pi^B_{k'}]=0$. This
implies that the measurement $\{\Pi^B_k\}$ are eigenprojectors of
$\rho^B$. Another proof of this fact was recently given by
Luo~\cite{LuoPRAb}. On the other hand, if $\rho^B$ is non-degenerate
with spectral decomposition
$\rho^B=\sum_k\lambda_k|\psi_k\rangle\langle \psi_k|$, then we know
that the dimension of the eigenspace of $\rho^{B}$ corresponding to
the eigenvaule $\lambda_k$ is one. This implies that $\{|\psi_k\rangle\langle \psi_k|\}$ is the only von Neumann measurement
that does not
disturb $\rho^B$. Thus, the maximum in Eq.~(\ref{eq:REN}) is not
necessary.

As for property (iv), for any bipartite state $\rho^{AB}$, the relative entropy of entanglement~\cite{VedralPRAb} is defined as
\begin{equation}
E_{\mathrm{RE}}(\rho^{AB})=\min_{\sigma^{AB}} S\big(\rho^{AB}||\sigma^{AB}\big),
\end{equation}
where $\sigma^{AB}$ is taken over all separable states. For any entangled
state $\rho^{AB}$, we obtain
\begin{equation}
E_{\mathrm{RE}}(\rho^{AB})>0.\label{eq:APP1}
\end{equation}
Since $\tilde{\rho}^{AB}$ is a separable state for any locally invariant measurement, we have
\begin{equation}
E_{\mathrm{RE}}(\rho^{AB})\leq S(\rho^{AB}||\tilde{\rho}^{AB})\leq \mathcal{N}^{\leftarrow}_{\mathrm{RE}}(\rho^{AB}).\label{eq:APP2}
\end{equation}
Substituting Eq.~(\ref{eq:APP1}) into Eq.~(\ref{eq:APP2}) leads to
the desired result.

In general, a von Neumann measurement induces entropy increase, and
this change is described by the relative entropy. As pointed out
in~\cite{Modi2010}, one checks that
\begin{equation}
S\big(\rho^{AB}||\tilde{\rho}^{AB}\big)=S(\tilde{\rho}^{AB})-S(\rho^{AB}).\label{eq:RE_Entropy}
\end{equation}
Here, $S(X)=-\mathrm{Tr}X\log_2X$ is von Neumann entropy~\cite{Nielsenbook}.
Thus, the relative entropy of nonlocality can be rewritten as
\begin{equation}
\mathcal{N}^{\leftarrow}_{\mathrm{RE}}(\rho^{AB})=\max_{\{\Pi^B_k\}}\Big[S(\tilde{\rho}^{AB})-S(\rho^{AB})\Big].\label{eq:difference of entropy}
\end{equation}
This implies that the relative entropy of nonlocality is the maximal entropy increase with respect to the locally invariant measurements.
This result gives a meaningful physical interpretation of the measurement-induced nonlocality. More importantly, it allows us to derive a upper bound for the relative entropy of nonlocality.
Next, we will present a formal proof for the upper bound of relative entropy of nonlocality.
\begin{theorem}
 For any bipartite state $\rho^{AB}$, the relative entropy of nonlocality cannot exceed the entropy of measured system,
 \begin{equation}
 \mathcal{N}^{\leftarrow}_{\mathrm{RE}}(\rho^{AB})\leq S(\rho^B).\label{eq:REN_UB}
 \end{equation}
\end{theorem}
\begin{proof}
For the local von Neumann measurements~$\{\Pi^B_k\}$ which leaves
$\rho^B$ invariant, $\rho^B=\sum_kp_k\Pi^B_k$ is a spectral
decomposition of $\rho^B$, and $p_k=\mathrm{Tr}(\rho^B\Pi^B_k)$ is
the probability of the outcome $k$~\cite{LuoEPL,LuoPRL}. In such a
case, for the post-measurement state $\tilde{\rho}^{AB}$, from the
result in~\cite{Nielsenbook}, one can directly obtains the following
equation
\begin{equation}
S(\tilde{\rho}^{AB})=S(\rho^B)+\sum_kp_kS(\rho^A_k).
\end{equation}
Substituting this expression into Eq.~(\ref{eq:difference of
entropy}), we get an explicit formula for the relative entropy of
nonlocality,
\begin{equation}
\mathcal{N}^{\leftarrow}_{\mathrm{RE}}(\rho^{AB})=S(\rho^B)-\Big[S(\rho^{AB})
-\max_{\{\Pi_{k}^{B}\}}\sum_kp_kS(\rho^A_k)\Big].\label{eq:REN_3}
\end{equation}
From the result in~\cite{ZJXiJPA}, for any von Neumann measurements, one checks that
 \begin{equation}
 \sum_kp_kS(\rho^A_k)\leq S(\rho^{AB}).
 \end{equation}
Since this order relation is true for all von Neumann measurements, combining this inequality with Eq.~(\ref{eq:REN_3}), we get
\begin{equation}
\mathcal{N}^{\leftarrow}_{\mathrm{RE}}(\rho^{AB})\leq S(\rho^B).\label{eq:REN_upper bound}
\end{equation}
\end{proof}
In particular, let $\rho^{AB}=|\psi\rangle^{AB}\langle\psi|$ be a bipartite pure state, then the relative entropy of nonlocality is equal to the entropy of the measured system,
\begin{equation}
\mathcal{N}^{\leftarrow}_{\mathrm{RE}}(|\psi\rangle^{AB}\langle\psi|)=S(\rho^{B}),\label{eq:REN_von Neumann entropy}
\end{equation}
which is a direct consequence of Eq.~(\ref{eq:REN_3}).
For bipartite pure state $|\psi\rangle^{AB}$, every post-measurement state of the subsystem $A$ is a pure state for
any von Neumann measurement on subsystem $B$, which implies that $S(\rho^{A}_k)=0$.
This shows that the relative entropy of nonlocality reduces to the entropy of the measured system entropy for pure states,
 which coincides with the relative entropy of quantumness and entanglement, but this is not true for general mixed states.

From Eq.~(\ref{eq:REN_3}),
the relative entropy of nonlocality~(\ref{eq:difference of entropy}) is also equal to
the maximal value of difference of two conditional entropies~\cite{ZJXiJPA}
\begin{equation}
\mathcal{N}^{\leftarrow}_{\mathrm{RE}}(\rho^{AB})=\max_{\{\Pi_{k}^{B}\}}\Big[S_{\{\Pi_{k}^{B}\}}(A|B)-S(A|B)\Big],\label{eq:difference of conditional entropy}
\end{equation}
where $S(A|B):=S(\rho^{AB})-S(\rho^{B})$ is the von Neumann conditional entropy and $S_{\{\Pi_{k}^{B}\}}(A|B):=\sum_kp_kS(\rho^A_k)$ is
quantum conditional entropy~\cite{Ollivier02,ZurekPRA,ZJXiJPA}. It is known that quantum discord~\cite{Ollivier02} is equal to the minimal difference of two conditional entropies for any von Neumann measurement~\cite{ZJXiJPA}.

 \subsection{Bell-diagonal states}
 One knows that the maximum is not necessary when the measured subsystem is non-degenerate, otherwise it is necessary to find a set of locally invariant measurement such that the maximal value of the first term in Eq.~(\ref{eq:difference of entropy}) is achieved.
Now, the main difficulty in calculating the relative entropy of nonlocality is to find an optimal locally invariant measurement for the degenerated measured subsystem. If the measured subsystem is degenerate, then
a number of eigenstates share a common eigenvalue.
These eigenstates can span the eigenspace of this eigenvalue. The linear combination of these eigenstates are still eigenstates in this eigenspace.
Then, we conclude that the maximal process is only obtained by taking maximal values in every degenerated eigenspace of the measured subsystem.

To give an intuitive understanding of the relative entropy of nonlocality, let us illustrate it
by a fundamental example.
We consider the Bell-diagonal states, whose reduced states are degenerate, namely,
\begin{equation}
\rho^{AB}=\frac{1}{4}\Big(I^A\otimes I^B+\sum_{i=1}^3c_i\sigma_i\otimes\sigma_i\Big),\label{eq:bell-diagonal state}
\end{equation}
where $I^{A(B)}$ is the identity operator on the subsystem $A(B)$, $c_i$ are real numbers and $\sigma_i$
are Pauli operators. To obtain the relative entropy of nonlocality, combined with Eq.~(\ref{eq:difference of conditional entropy}), we only evaluate the maximal quantum conditional entropy under the locally invariant measurements on $B$. Following the approach in~\cite{LuoPRAb}, we have
\begin{equation}
\max S_{\{\Pi_{k}^{B}\}}(A|B)=f(c_{\min}),
\end{equation}
where the maximum is taken over all the von Neumann measurements on $B$ (since $\rho^{B}=\frac{I^B}{2}$ is degenerate, any von Neumann measurement will leave $\rho^{B}$ invariant). Here, the function
\begin{equation}
f(x)=-\frac{1+x}{2}\log_2\frac{1+x}{2}-\frac{1-x}{2}\log_2\frac{1-x}{2}
 \end{equation}
 and $c_{\min}=\min\{|c_1|,|c_2|,|c_3|\}$. Then, the relative entropy of nonlocality can be given as
\begin{align}
\mathcal{N}^{\leftarrow}_{\mathrm{RE}}(\rho^{AB})&=\max S_{\{\Pi_{k}^{B}\}}(A|B)-S(A|B)\nonumber\\
&=f(c_{\min})-f(c_1,c_2,c_3),
\end{align}
where
\begin{align}
&f(c_1,c_2,c_3)\nonumber\\
=&-\Big(\frac{1-c_1-c_2-c_3}{4}\log_2\frac{1-c_1-c_2-c_3}{4}\nonumber\\
&+\frac{1-c_1+c_2+c_3}{4}\log_2\frac{1-c_1+c_2+c_3}{4}\nonumber\\
&+\frac{1+c_1-c_2+c_3}{4}\log_2\frac{1+c_1-c_2+c_3}{4}\nonumber\\
&+\frac{1+c_1+c_2-c_3}{4}\log_2\frac{1+c_1+c_2-c_3}{4}\Big)-1.\nonumber
\end{align}

 In particular, taking $|c_1|=1, 0<|c_2|=|c_3|=c<1 $, one has $c_{\min}=c$. In such a case, we have
 \begin{align}
 \mathcal{N}^{\leftarrow}_{\mathrm{RE}}(\rho^{AB})=1.
 \end{align}
From this example, we find that the relative entropy of nonlocality is equal to the von Neumann entropy of the measured subsystem for some two-qubit mixed states.

\section{The relations with other measures}\label{sec:REN_HSN}
\subsection{The relative entropy of nonlocality and the geometric nonlocality}

We will compare $\mathcal{N}^{\leftarrow}_{\mathrm{G}}$ and $\mathcal{N}^{\leftarrow}_{\mathrm{RE}}$, the latter will be shown to majorize the former. The main result of this section is shown in the following.
\begin{theorem}
For any bipartite state $\rho$, the relative
entropy of nonlocality is always greater than or equal to the square
of the geometric nonlocality,
\begin{equation}
\frac{1}{2\ln2}\mathcal{N}^{\leftarrow}_{\mathrm{G}}(\rho)^2\leq \mathcal{N}^{\leftarrow}_{\mathrm{RE}}(\rho).
\end{equation}
\end{theorem}
\begin{proof}
If there exists an optimal locally invariant measurement $\{{\Pi^*_k}^B\}$ such that $\tilde{\rho}^*=\sum_kI^A\otimes{\Pi^*_k}^B\rho I^A\otimes{\Pi^*_k}^B$, then the closest $\rho^*$ achieves the maximum of the Hilbert-Schmidt norm $||\rho-\tilde{\rho}||$. Thus, we have
\begin{equation}
\mathcal{N}^{\leftarrow}_{\mathrm{G}}(\rho)^2=||\rho-\tilde{\rho}^*||^2.\label{eq:HSN1}
\end{equation}
We know from~\cite{Watrous} that for any two quantum states $\rho$
and $\sigma$, the following inequality holds,
\begin{equation}
\frac{1}{2\ln2}||\rho-\sigma||^2\leq S(\rho||\sigma).\label{eq:REvsHS}
\end{equation}
Thanks to Eq.(\ref{eq:REvsHS}), this yields
\begin{equation}
\frac{1}{2\ln2}||\rho-\tilde{\rho}^*||^2\leq S(\rho||\tilde{\rho}^*).\label{eq:REN_HSN}
\end{equation}
From the definition of the relative entropy of nonlocality, we have
\begin{equation}
S(\rho||\tilde{\rho}^*)\leq \mathcal{N}^{\leftarrow}_{\mathrm{RE}}(\rho).\label{eq:inequ_REN}
\end{equation}
Combined Eqs.~(\ref{eq:HSN1}), (\ref{eq:REN_HSN}) with Eq.~(\ref{eq:inequ_REN}), we get the desired result.
\end{proof}

\subsection{The relative entropy of nonlocality and distillable entanglement}

As pointed out in~\cite{Streltsov11,Piani11}, the von Neumann
measurement on a part of a composite quantum system unavoidably
creates distillable entanglement between the apparatus $M$ and the
system $AB$. Here, the von Neumann measurement on $B$ can be
realized by a unitary on the total system $M$ and $AB$. The
entanglement between $M$ and $AB$ is called the entanglement created
in the von Neumann measurement $\{\Pi^B_k\}$ on $B$. From
Fig.~\ref{fig:relation}, a local von Neumann measurement can be
represented as
 \begin{equation}
 \rho^{AB}\rightarrow\rho^{AB}\otimes|0\rangle^{M}\langle 0|\rightarrow\tilde{\rho}^{ABM}\rightarrow\tilde{\rho}^{AB}.
 \end{equation}
Suppose that the system $AB$ and the apparatus $M$ is initially in a
product state, namely,
\begin{equation}
\rho^{ABM}=\rho^{AB}\otimes |0\rangle^{M}\langle 0|.
\end{equation}
After the action of the unitary $I^A\otimes U^{BM}$ the state becomes,
\begin{equation}
\tilde{\rho}^{ABM}=\big(I^A\otimes U^{BM}\big)\rho^{ABM}\big(I^A\otimes {U^{\dagger}}^{BM}\big).
 \end{equation}
We then discard the apparatus $M$, leaving system $AB$ in the post-measurement state
 \begin{equation}
\tilde{\rho}^{AB}=\mathrm{Tr}_{M}\tilde{\rho}^{ABM}.
 \end{equation}
 Consider the distillable entanglement $E_D$~\cite{Plenio07,Bennett96}, for any locally invariant measurements $\{\Pi^{B}_k\}$ on $B$, the following equality holds~\cite{Streltsov11}
\begin{equation}
E^{AB|M}_D(\tilde{\rho}^{ABM})=S\big(\rho^{AB}||\tilde{\rho}^{AB}\big).
\end{equation}
Then, we maximize the equation over all the locally invariant measurements on $B$ and obtain
\begin{equation}
\mathcal{N}^{\leftarrow}_{\mathrm{RE}}(\rho^{AB})=\max_{\{\Pi^B_k\}}E^{AB|M}_D(\tilde{\rho}^{ABM}).
\end{equation}
This result shows that the relative entropy of nonlocality is equal to the maximum of the distillable entanglement between the system and the measurement apparatus.

\begin{figure}
  \includegraphics[scale=0.5]{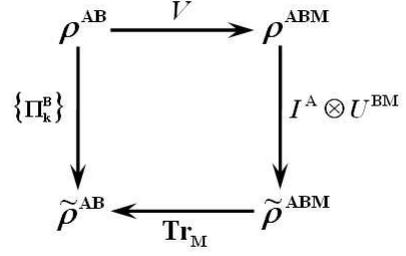}
  \caption{This figure shows an equivalent form. Here, $V$ is
  an isometry, $V: \mathcal{H}^{AB}\rightarrow\mathcal{H}^{AB}\otimes\mathcal{H}^{M}$, with $V:=I^{AB}\otimes |0\rangle^M$ and $V^\dagger V=I^{AB}$. The connection is explained in the text.\label{fig:relation}}
\end{figure}

\subsection{Trade-off relations}
Consider performing the von Neumann measurement $\{\Pi^B_k\}$ which does not disturb $\rho^B$ locally, one can obtain an ensemble of the subsystem $A$, i.e., $\{p_k,\rho^{A}_k\}$.
From the results \cite{DevetakPRA,RenesPRL}, we know that the locally accessible mutual information $\chi(\{p_k,\rho^{A}_k\})$ is also called quantum side information, namely,
\begin{equation}
\chi\left(\{p_k,\rho^{A}_k\}\right):=S(\rho^{A})-\sum_ip_kS(\rho^{A}_k).
\end{equation}
If the maximum is taken over all the von Neumann measurements $\{\Pi^B_k\}$ which do not disturb $\rho^B$ locally, then we can define the minimal quantum side information, namely,
\begin{align}
\mathcal{S}^{\leftarrow}_\chi(\rho^{AB}):&=\min_{\{\Pi^B_k\}}\chi\left(\{p_k,\rho^{A}_k\}\right)\nonumber\\
&=S(\rho^{A})-\max_{\{\Pi^B_k\}}\sum_ip_kS(\rho^{A}_k)
\end{align}

Therefore, for any bipartite quantum state $\rho^{AB}$, the sum of the relative entropy of nonlocality and the minimal quantum side information is equal to quantum mutual information, namely,
\begin{equation}
\mathcal{N}^{\leftarrow}_{\mathrm{RE}}(\rho^{AB})+\mathcal{S}^{\leftarrow}_\chi(\rho^{AB})=\mathcal{I}(\rho^{AB}),
\end{equation}
where $\mathcal{I}(\rho^{AB})=S(\rho^A)+S(\rho^B)-S(\rho^{AB})$ is quantum mutual information~\cite{Nielsenbook}.

From the result in~\cite{Coles1664v2}, we will derive some trade-off
relations in terms of this relation. We consider a tripartite pure
state $\rho^{ABC}$ such that $\rho^{AB}=\mathrm{Tr}_C(\rho^{ABC})$.
Performing the von Neumann measurement $\{\Pi^B_k\}$ which does not
disturb $\rho^B$ locally, we have
\begin{align}
S(\rho^{AB}||\tilde{\rho}^{AB})&=S(\tilde{\rho}^{AB})-S(\rho^{AB})\nonumber\\
&=S(\rho^{B})+\sum_kp_kS(\rho^{A}_k)-S(\rho^{C})\nonumber\\
&=S(\rho^{B})+\sum_kp_kS(\rho^{C}_k)-S(\rho^{C})\nonumber\\
&=S(\rho^{B})-\chi\left(\{p_k,\rho^{C}_k\}\right).
\end{align}
The third equality is due to the fact that for a tripartite pure
state $\rho^{ABC}$, after performing the von Neumann measurement
$\{\Pi^B_k\}$ on subsystem $B$, the residual state of subsystem $AC$
is still a pure state $\rho^{AC}_k$, and one obtains
$S(\rho^{A}_k)=S(\rho^{C}_k)$ for every $k$.

We then maximize the equation over all the locally invariant measurements on $B$, and obtain a trade-off relation as following
\begin{align}
\mathcal{N}^{\leftarrow}_{\mathrm{RE}}(\rho^{AB})+\mathcal{S}^{\leftarrow}_\chi(\rho^{CB})=S(\rho^{B}).
\end{align}
This equality shows that the amount of nonlocality
between $A$ and $B$, plus the amount of minimal quantum side information between
$B$ and the complementary part $C$, must be equal to the entropy
of the measured subsystem $B$.
For general tripartite mixed
state $\rho^{ABC}$, this equality is not true. To be convinced, let us
purify mixed state $\rho^{ABC}$ as $\rho^{ABC}=\mathrm{Tr}_{D}|\Psi\rangle^{ABCD}\langle\Psi|$,
then we have
\begin{align}
\mathcal{N}^{\leftarrow}_{\mathrm{RE}}(\rho^{AB})+\mathcal{S}^{\leftarrow}_\chi(\rho^{(CD)B})=S(\rho^{B}).
\end{align}
Since discarding quantum systems never increases quantum side information (Holevo information)~\cite{ColesPRA,SchumacherPRA}, then we have
\begin{equation}
\mathcal{S}^{\leftarrow}_\chi(\rho^{CB})\leq
\mathcal{S}^{\leftarrow}_\chi(\rho^{(CD)B}),
\end{equation}
which is equivalent to
\begin{equation}
\mathcal{N}^{\leftarrow}_{\mathrm{RE}}(\rho^{AB})+\mathcal{S}^{\leftarrow}_\chi(\rho^{CB})\leq S(\rho^{B})
.\label{eq:tradeoff_mixed_states}
\end{equation}

Using quantum side information, the missing information \cite{DevetakPRA,ColesPRA} about the locally invariant measurements on $B$ given the subsystem $A$ can be defined as
\begin{equation}
S({\{\Pi^B_k\}}|A):=S(\rho^{B})-\chi\left(\{p_k,\rho^{A}_k\}\right).
\end{equation}
This quantity is a measure of absence of the locally invariant
measurements on $B$ of information form $A$. Since
$\chi\left(\{p_k,\rho^{A}_k\}\right)\leq \min[S(\rho^{A}),
S(\rho^{B})]$~\cite{ColesPRA}, we then have $ S({\{\Pi^B_k\}}|A)\geq
0$. We maximize the equation over all the locally invariant
measurements on $B$, and find that the relative entropy of
nonlocality is equal to the maximal missing information, namely,
\begin{equation}
\mathcal{N}^{\leftarrow}_{\mathrm{RE}}(\rho^{AB})=\max_{\{\Pi^B_k\}}S({\{\Pi^B_k\}}|A).
\end{equation}

\section{conclusion}\label{sec:conclusion}

In conclusion, we have quantified the measurement-induced nonlocality in terms of the relative entropy, which is called the relative entropy of nonlocality. It can be interpreted as the maximal entropy increase after the locally invariant measurements.
We have investigated its properties, and presented an upper bound of relative entropy of nonlocality for arbitrary quantum states. For pure states, the upper bound is saturated, and it is identical with the corresponding measures of entanglement and quantum correlations.

We have shown that the relative entropy of nonlocality is always
greater than or equal to the square of the geometric nonlocality for
arbitrary quantum states. We have further shown that the relative
entropy of nonlocality is equal to the maximal distillable
entanglement between the measurement apparatus and the system if the
von Neumann measurement is performed on one part of the system. Some
trade-off relations have been derived by the quantum side
information and the missing information. In particular, we have
obtained that the sum of the relative entropy of nonlocality and the
minimal quantum side information exactly equals quantum mutual
information for bipartite quantum states. If one adopts Streltsov's
and Piani's suggestions~\cite{Streltsov11,Piani11}, our approach can
be generalized to the multipartite setting.

\begin{acknowledgments}
The authors are very grateful to the referees for helpful comments
and criticisms. We thank X.-M. Lu for interesting discussions. Z. J.
Xi is supported by the Superior Dissertation Foundation of Shaanxi
Normal University (S2009YB03). Y. M. Li is supported by NSFC with
Grant No.60873119, and the Higher School Doctoral Subject Foundation
of Ministry of Education of China with Grant No.200807180005. X.Wang
is supported by NFRPC with Grant No. 2012CB921602, and NSFC with
Grants No. 11025527 and No. 10935010.
\end{acknowledgments}

\end{document}